\documentclass{llncs}
\usepackage{amssymb}
\usepackage{amsmath}
\usepackage{amsxtra}
\usepackage{graphicx}
\usepackage{graphics}
\usepackage[svgnames]{xcolor} 
\usepackage{color}
\usepackage[pdftex,pagebackref=false]{hyperref}
\hypersetup{colorlinks,citecolor=darkblue,linkcolor=darkblue,urlcolor=darkblue}
\definecolor{darkblue}{rgb}{0,0.08,0.45}

\newcommand{\full}[2]{#1}

\full{\usepackage{fullpage}}{}

\newcommand{\treew}{tree-width}

\newcommand{\minorfr}{minor-free}

\newcommand{\dist}{\text{dist}}
\newcommand{\ceil}[1]{\lceil #1 \rceil}

\newcommand{\abs}[1]{\left| #1 \right|}
\newcommand{\comment}[1]{}
\newcommand{\set}[1]{\{#1\}}

\newcommand{\invackermann}{\alpha}

\newcommand{\Patrascu}{P{\v a}tra{\c s}cu}

\def\genusprepro{O(n(\lg n)(g^3+\lg n))}
\def\genusaltprepro{O(n(\lg n)(g/\epsilon+\lg n))}
\def\genusfastqprepro{O(n(\lg n)^3\epsilon^{-2}+n(\lg n)g/\epsilon)}

\newcommand{\last}{\text{last}}
\newcommand{\used}{\mu}
\newcommand{\length}{\mbox{length}}

\begin{document}

\title{Linear-Space Approximate Distance Oracles for\\ Planar, Bounded-Genus\full{,}{} and Minor-Free Graphs\full{\thanks{An extended abstract is to appear in the Proceedings of the 38th International Colloquium on Automata, Languages and Programming (ICALP 2011)}}{\thanks{An extended version can be found online~\cite{TR}.}}}
\author{Ken-ichi Kawarabayashi \inst{1}
\and Philip N.~Klein\inst{2}
\and Christian Sommer\inst{3}
} \institute{NII, Tokyo, Japan
\and Brown U, Providence RI
\and MIT, Cambridge MA} \maketitle

\begin{abstract}

  A $(1+\epsilon)$--approximate distance oracle for a graph is a data
  structure that supports approximate point-to-point
  shortest-path-distance queries.  The most relevant measures for a
  distance-oracle construction are: space, query time, and
  preprocessing time.

There are strong distance-oracle constructions known for planar
  graphs (Thorup, JACM'04) and, subsequently, minor-excluded graphs (Abraham
  and Gavoille, PODC'06).  However, these require $\Omega(\epsilon^{-1} n \lg
  n)$ space for $n$--node graphs.  

\full{
  We argue that a very low space requirement is essential.  Since
  modern computer architectures involve hierarchical memory (caches,
  primary memory, secondary memory), a high memory requirement in
  effect may greatly increase the actual running time.  Moreover, we
  would like data structures that can be deployed on small mobile
  devices, such as handhelds, which have relatively small primary
  memory.  
}{}

  In this paper,  for planar graphs, bounded-genus graphs, and
  minor-excluded graphs we give distance-oracle constructions that
  require only 
  $O(n)$ space.  The big $O$ hides only a fixed constant, independent of
  $\epsilon$ and independent of genus or size of an excluded minor.
  The preprocessing times for our distance oracle are also faster
  than those for the previously known constructions.  For planar
  graphs, the preprocessing time is $O(n \lg^2 n)$.
  However, our constructions have slower query times.  For
  planar graphs, the query time is $O(\epsilon^{-2} \lg^2 n)$.  

\full{
  For bounded-genus graphs, there was previously no 
  distance-oracle construction known other than the one implied by the
  minor-excluded construction, for which the constant is enormous and
  the preprocessing time is a high-degree polynomial.  In our result,
  the query time is $O(\epsilon^{-2}(\lg n + g)^2)$ and the preprocessing 
  time is $\genusprepro$.
 }{} 

  For all \full{these}{our} linear-space results, we can in fact ensure, for any
 $\delta>0$, that the space required is only $1+\delta$ times the
 space required just to represent the graph itself.

\end{abstract}

\full{\newpage}{}

\section{Introduction}

A {\em $(1+\epsilon)$--approximate distance oracle} for a graph is a data structure
that supports point-to-point approximate distance queries.  A
distance-oracle construction for a family of graphs has three
complexity measures:
\begin{itemize}
\item {\em preprocessing time:} time to build the data structure,
\item {\em space:} how much space is occupied by the data structure, and
\item {\em query time:} how long does it take for a query to be
  answered.
\end{itemize}
Each of these quantities might depend on the {\em stretch} parameter
$1+\epsilon$ (which is defined as the maximum ratio over all pairs of 
nodes of the query output divided by the length of a shortest path) 
as well as the size of the graph.

\full{\paragraph{General graphs}}{}
For general graphs, for stretch less than~2, no approximate
distance oracle is known that achieves subquadratic space and
sublinear query time.  \full{(In Section~\ref{sec:general-graphs}, we
briefly survey work on general graphs.)  }{}

\full{\paragraph{Restricted graph families}}{}
The only known constructions
that achieve $(1+\epsilon)$ stretch are for restricted families of graphs:
planar graphs~\cite{ThorupJACM04}, minor-excluded
graphs~\cite{AbrahamG06}, and graphs of low doubling
dimension~\cite{conf/stoc/Talwar04,MendelHarPeled,SlivkinsPODC05JournalVersion,journals/corr/abs-1008-1480}.  
\full{(In Section~\ref{sec:previous-work}, we
survey previously known results in this area.)}{}

\full{
Fortunately, such graphs arise in applications, e.g. relating to road
maps.  Distance oracles can be used in finding nearby points of
interest, in navigation and route-planning, and in algorithms for
solving other optimization problems such as vehicle routing.  
}{}

\full{\paragraph{Space requirements}}{}
One obstacle to the widespread adoption of this technique may have
been the space requirements of known distance oracles.  Even the most
compact distance oracle of Thorup~\cite{ThorupJACM04} requires $\Omega(\epsilon^{-1} n \lg
  n)$ space for $n$--node graphs.  Even though the constant is quite
  modest, the storage required is rather large~\cite{distance-oracle-experiment}.

\full{
Since modern computer architectures involve hierarchical memory (caches,
primary memory, secondary memory), a high memory requirement in
effect may greatly increase the actual query time.  If the distance
oracle could fit in cache, the query time could be much faster than
if secondary or primary memory must be accessed.  For smaller, less
powerful mobile devices such as handhelds, the problem is exacerbated.}{}

\subsection*{Our contribution}
In this paper, for every family of graphs for which a nontrivial
$(1+\epsilon)$--approximate distance oracle is known (planar, bounded genus, 
$H$--\minorfr, bounded doubling dimension), we give such a
distance oracle that in addition requires only linear space.\footnote{
Our results for bounded-doubling-dimension graphs hold only for unit 
lengths.}  In fact,
for any $\delta >0$, there is such a distance oracle whose space
requirement is only $1+\delta$ times the space required just to store
the graph itself; thus the overhead due to the distance oracle is in
essence negligible.

We achieve this while increasing the query time by a factor that is
almost proportional to the decrease in space.\footnote{The product of 
space times query time for our oracle is $O(n\epsilon^{-2}\lg^2n)$ while
that same product is $O(n\epsilon^{-2}\lg n)$ for~\cite{ThorupJACM04}} 
For planar graphs, the
query time of our oracle is $O(\epsilon^{-2}\lg^2n)$.  Although
the query time for our constructions is slower than that for the
superlinear-space constructions, the increase may be partly
made up for by the decrease in actual time due to better memory
performance (because of the memory hierarchy).  The space/query-time
tradeoff is tunable, so the construction can be adapted to a
particular architecture.

\full{The preprocessing time is also faster for our schemes than for the
superlinear-space constructions.}{}

\full{\paragraph{Bounded-genus graphs}}{}
For bounded-genus graphs, there was
previously no distance-oracle construction known other than that
implied by the minor-excluded construction, for which the constant is
enormous and the preprocessing time is a high-degree polynomial.  We
give a more efficient construction tailored to graphs of genus~$g$.    \full{For
our linear-space oracle, the query time is $O(\epsilon^{-2}(\lg n + g)^2)$ 
and the preprocessing time is $\genusprepro$. (There is also an alternative construction 
with preprocessing time $\genusaltprepro$.)
We also provide an approximate distance oracle using  space $O(n\epsilon^{-1}(\lg n+g))$
but having  faster query time $O(g/\epsilon)$ (Theorem~\ref{thm:genusfastquerythm}).}{}

\full{
\paragraph{Bounded-doubling-dimension graphs}
We also provide a linear-space approximate distance oracle for unit-length graphs with 
bounded doubling dimension (Theorem~\ref{thm:doublingthm}) using an approach that 
does not require separators. The query time of our construction is 
$\epsilon^{-O(\alpha^2)} \cdot (\lg n)^{O(\alpha)}$. 
}{}

\full{\paragraph{Summary of our results}}{}
A summary of our results is given in Table~\ref{tab:our-results}. 

\begin{table*}[h!]
\begin{center}
\begin{tabular}{|l  | l | l | l | l|}
\hline
Graph Class & Preprocessing & Query &  \\
\hline
\hline
Planar Undirected & $O(n\lg^2n)$ & $O(\epsilon^{-2}(\lg n)^2)$ & Theorem~\ref{thm:main}\\
Planar Directed & $O(n(\lg(nN))(\lg n)^3\epsilon^{-2})$ & $O((\epsilon^{-1}(\lg n)(\lg(nN)))^{2})$ & \full{Section~\ref{sec:planar:extensions}}{\cite{TR}} \\
$\ \ \ $Reachability Oracle &$O(n\lg n)$ & $O(\lg^2n)$ &\full{Section~\ref{sec:planar:extensions}}{\cite{TR}} \\
Genus~$g$ & $\genusprepro$ & $O(\epsilon^{-2}(\lg n + g)^2)$ & Theorem~\ref{thm:genusthm}\\
$H$--minor-free & $O(poly(n,\epsilon))$ & $O(\epsilon^{-2}(\lg n)^2)$ & \full{Theorem~\ref{thm:minorthm}}{\cite{TR}}\\
$\alpha$--doubling, unit lengths & $\epsilon^{-O(\alpha)}O(poly(n))$ & $\epsilon^{-O(\alpha^2)} \cdot (\lg n)^{O(\alpha)}$ & \full{Theorem~\ref{thm:doublingthm}}{\cite{TR}}\\
\hline
\end{tabular}
\end{center}
\caption{Time complexities of our linear-space $(1+\epsilon)$--approximate distance oracles. $N$~denotes the largest integer weight.
Due to space restrictions, some are in~\cite{TR}. 
\label{tab:our-results}}
\end{table*}

\section{Previous work on approximate distance oracles}

\full{\subsection{General and sparse graphs}}{\paragraph{General and sparse graphs}}
\label{sec:general-graphs}
Thorup and Zwick~\cite{ThorupZwick2005} gave asymptotically almost
optimal trade-offs for distance oracles for general undirected graphs, 
proving that for any graph and for any integer $k$ there is a $(2k-1)$--approximate
distance oracle using space $O(kn^{1+1/k})$ and query time
$O(k)$. They also prove that, if stretch strictly less than $2k+1$ is
desired, then $\Omega(n^{1+1/k})$~bits of space are necessary. A
slightly weaker lower bound holds for sparse
graphs: Sommer, Verbin, and Yu~\cite{SparseDO} prove that a distance oracle with
stretch $k$ and query time $t$ requires space $n^{1+\Omega(1/(kt))}$ (up
to poly-logarithmic factors). \full{Tight with respect to this bound, the distance oracle 
of Mendel and Naor~\cite{MendelN07} has query
time $O(1)$ and stretch $O(k)$ using space $O(n^{1+1/k})$.}{} The oracle with 
the best stretch factor is by \Patrascu\ and Roditty~\cite{PatrascuRoditty}, who 
recently gave a 2--approximate distance oracle
using space $O(n^{5/3})$ on sparse graphs. Distance oracles with stretch 
strictly less than 2 have not been achieved for general graphs.

\full{\subsection{Restricted graph classes: planar, excluded-minor, 
 and bounded-doubling-dimension graphs}}{\paragraph{Restricted graph classes}}
 \label{sec:previous-work}
 
For restricted classes of graphs, better distance oracles 
are known and stretch $1+\epsilon$ can be achieved. 

\full{\paragraph{Planar graphs}}{}
Thorup~\cite{ThorupJACM04} presents efficient
$(1+\epsilon)$--approximate distance oracles for planar digraphs.
(There is a slight improvement~\cite{conf/soda/Klein05} to the
preprocessing time for one case.)
Table~\ref{tab:planardoresults} lists these results. 

There are also many results on exact distance oracles for planar
graphs. The best is that of Fakcharoenphol and Rao~\cite{journals/jcss/FakcharoenpholR06} 
and its subsequent improvements~\cite{conf/soda/Klein05,journals/talg/KleinMW10,esa/MozesW10} and 
variants~\cite{journals/corr/abs-1011-5549,Nussbaum10}.  
Faster per-query time can be achieved by using more 
space~\cite{conf/soda/Cabello06,journals/corr/abs-1011-5549}.  
However, all these
results require polynomial (but sublinear) query time.
There are also results on special cases of planar graphs and
special kinds of
queries~\cite{\full{conf/stacs/DjidjevPZ95,}{}journals/algorithmica/DjidjevPZ00\full{,conf/icalp/DjidjevPZ91}{},stoc/ChenX00,journals/talg/KowalikK06}.

\begin{table*}[h!]
\begin{center}
\begin{tabular}{\full{|l}{}  | l | l | l | l|}
\hline
\full{&}{} Preprocessing & Space & Query & Reference \\
\hline
\hline
\full{Directed &}{}\footnotesize \full{$O(n(\lg(nN))(\lg n)^3\epsilon^{-2})$}{$O(n\lg^4n\epsilon^{-2})$}  & \footnotesize \full{$O\left(n\cdot\epsilon^{-1}(\lg n)(\lg(nN))\right)$}{$O(n\epsilon^{-1}\lg^2n)$}           & \footnotesize\full{$O(\lg\lg(nN)+\epsilon^{-1})$}{$O(\lg\lg(n)+\epsilon^{-1})$} & \cite[Thm.~3.16]{ThorupJACM04}\\
\full{Directed &}{}\footnotesize \full{$O(n(\lg n)^2(\lg(nN))\epsilon^{-1})$}{$O(n\lg^3 n\epsilon^{-1})$}  &\footnotesize \full{$O\left(n\cdot\epsilon^{-1}(\lg n)(\lg(nN))\right)$}{$O(n\epsilon^{-1}\lg^2 n)$}           & \footnotesize\full{$O((\lg n)(\lg\lg(nN)+\epsilon^{-1}))$}{$O(\lg n(\lg\lg n+\epsilon^{-1}))$} & \cite[Prop.~3.14]{ThorupJACM04}\\
\full{Directed &}{}\footnotesize \full{$O(n(\lg n+\epsilon^{-1})(\lg n)(\lg(nN)))$}{$O(n\lg^2n(\lg n+\epsilon^{-1}))$}  &\footnotesize\full{$O\left(n\cdot\epsilon^{-1}(\lg n)(\lg(nN))\right)$}{$O(n\epsilon^{-1}\lg^2 n)$}           & \footnotesize\full{$O((\lg n)(\lg\lg(nN)+\epsilon^{-1}))$}{$O(\lg n(\lg\lg n+\epsilon^{-1}))$} & \cite[Sec.~7]{conf/soda/Klein05}\\
\hline
\full{Undirected &}{}\footnotesize $O(n(\lg n)^3\epsilon^{-2})$  &\footnotesize \full{$O\left(n\cdot\epsilon^{-1}\lg n\right)$}{$O(n\epsilon^{-1}\lg n)$}           &\footnotesize $O(\epsilon^{-1})$ & \cite[Thm.~3.19]{ThorupJACM04}\\
\full{Undirected &}{}\footnotesize $O(n(\lg n)^2\epsilon^{-1})$  & \footnotesize\full{$O\left(n\cdot\epsilon^{-1}\lg n\right)$}{$O(n\epsilon^{-1}\lg n)$}           & \footnotesize$O(\epsilon^{-1}\lg n)$ & \cite[Implicit]{ThorupJACM04}\\
\hline
\end{tabular}
\end{center}
\caption{Time and space complexities of $(1+\epsilon)$--approximate distance oracles for planar graphs on $n$ nodes. \full{$N$~denotes the largest integer weight.}{In this table, the largest integer weight is assumed to be polynomial in~$n$. The upper part lists results for directed, the lower part lists results for undirected planar graphs.} \label{tab:planardoresults}}
\end{table*}

\full{\paragraph{Excluded-minor graphs}}{}

Abraham and Gavoille~\cite{AbrahamG06} extend Thorup's result to minor-free graphs. After a polynomial-time preprocessing step, point-to-point queries can be answered in time 
$O(\epsilon^{-1}\lg n)$ using a data structure of size $O( n\epsilon^{-1}\lg n)$. 

\full{
\paragraph{Bounded-\treew\ graphs}

For digraphs with \treew~$w$, Chaudhuri and Zaroliagis~\cite{journals/algorithmica/ChaudhuriZ00} give a $O(w^{3}n)$--time algorithm
to compute a distance oracle with query time $O(w^{3}\invackermann(n))$, where $\invackermann(n)$ denotes the inverse Ackermann function. 
Gavoille et al.~\cite[Theorem~2.4]{journals/jal/GavoillePPR04}
provide a distance oracle with space $O(n\cdot w \lg^2n)$ and query time $O(\lg n)$.
}{}

\full{
\paragraph{Graphs of bounded doubling dimension}

Let $\Delta$ denote the aspect ratio (diameter divided by minimum distance) and let $\alpha=\lg_2\lambda$ 
denote the doubling dimension. 
Har-Peled and Mendel~\cite{MendelHarPeled}, improving upon earlier results by 
Talwar~\cite{conf/stoc/Talwar04} and Slivkins~\cite{SlivkinsPODC05JournalVersion}, 
provide a $(1+\epsilon)$--approximate distance oracle (they term it {\em compact 
representation scheme}) using space $(1/\epsilon)^{O(\alpha)}n$ with query time 
$O(\alpha)$. 
Bartal et al.~\cite{journals/corr/abs-1008-1480}, extending~\cite{MendelHarPeled},
recently gave a distance oracle with constant query time, at the cost of increasing 
the space consumption. We shift the trade-off in the reverse direction, increasing the 
query time while reducing the space requirement to linear (see Section~\ref{sec:doubling} 
in the appendix for details). 
}{}

\def\mainthm{For any undirected planar graph $G$ with non-negative edge weights there exists a $(1+\epsilon)$--approximate distance oracle with query time
$O(\epsilon^{-2}\lg^2n)$, linear space, and preprocessing time $O(n\lg^2n)$.}

\def\genusthm{For any undirected graph $G$ embedded in a surface of Euler genus~$g$, there exists a $(1+\epsilon)$--approximate distance oracle with query time
$O(\epsilon^{-2}(\lg n + g)^2)$, linear space, and preprocessing time $\genusprepro$. The oracle can also be constructed in time $\genusaltprepro$.}

\def\genusthmfastquery{For any undirected graph $G$ embedded in a surface of Euler genus~$g$, there exists a $(1+\epsilon)$--approximate distance oracle with query time
$O(g/\epsilon)$, space $O(n(g+\lg n)/\epsilon)$, and preprocessing time $\genusfastqprepro$. The oracle can be distributed as a labeling scheme using $O((g+\lg n)/\epsilon)$~bits per node.}

\def\minorthm{For any minor $H$ there is an integer $h=h(H)$ such that 
for any undirected $H$--\minorfr\ graph $G$ with $n$ nodes and $m$ edges  there exists a
$(1+\epsilon)$--approximate distance oracle with query time
$O(h\epsilon^{-2}\lg^2n)$, space $O(m)$, and polynomial preprocessing time.}

\def\treewthm{For any graph $G$ on $n$ nodes and $m$ edges with \treew~$w=w(n)$
there is an exact distance oracle using space $O(m)$ with query time
$O(w^6\lg^4n)$. If $G$ is \minorfr\ or if $w = O(\lg n)$, 
the distance oracle can be constructed in polynomial
time.}

\def\doublingthm{For any unit-length graph $G=(V,E)$ on $n=\abs{V}$ nodes and
$m=\abs{E}$ edges with doubling dimension $\alpha=\lg_2\lambda$ and aspect
ratio $\Delta$ and for any $\epsilon > 0$, there exists a $(1 +
\epsilon)$--approximate distance oracle using space $O(m)$ and query
time $\lambda^{O(\lg ((\frac c\epsilon)^\alpha\frac
\alpha\epsilon\lg\frac\Delta\epsilon))}$ for some constant $c$.}

\def\doublingthm{For any unit-length graph $G=(V,E)$ on $n=\abs{V}$ nodes and
$m=\abs{E}$ edges with doubling dimension $\alpha=\lg_2\lambda$ and aspect
ratio $\Delta$ and for any $\epsilon > 0$, 
if there is a $(1+\epsilon)$--approximate distance labeling scheme with label 
length $\ell=\ell(n,\lambda,\Delta,\epsilon)$ and query time $q=\ell(n,\lambda,\Delta,\epsilon)$ 
then there exists a $(1 + \epsilon)$--approximate distance oracle using space $O(m)$ and query
time $\lambda^{O(\lg (\ell/\epsilon))}+q$. }

\def\doublingcor{
For any unit-length graph $G=(V,E)$ on $n=\abs{V}$ nodes and
$m=\abs{E}$ edges with doubling dimension~$\alpha$ and diameter $\Delta$, 
 and for any $\epsilon > 0$ 
there exists a $(1 + \epsilon)$--approximate distance oracle using space $O(m)$ and query
time $(\lg\Delta)^{O(\alpha)}\cdot(1/\epsilon)^{O(\alpha^2)}$. }

\full{\section{Linear-space approximate distance oracle for planar graphs}}{\section{Tunable approximate distance oracle for planar graphs}}
\label{sec:planar}
We prove our main theorem. The description of the improved preprocessing algorithm 
can be found in its own section (Section~\ref{sec:planar:prepro}). 
\begin{theorem}
 \mainthm
 \label{thm:main}
 \end{theorem}

\subsection{Review of Thorup's distance oracle}
\label{sec:planar:overview}

We briefly review a variant of Thorup's distance oracle for
undirected graphs (using somewhat different terminology).\footnote{This
variant does not appear in~\cite{ThorupJACM04} but is an obvious
simplification, analogous to that of~\cite[Proposition~3.14]{ThorupJACM04} (which
applies to directed graphs) resulting in slower query times.  Since
our query time is slower anyway, this simplified variant suffices.}

There are two core ideas.  The first is {\em approximately
  representing shortest paths that intersect a shortest path}.  Let
$P$ be a shortest path in a graph $G$.  A pair $(p,v)$ of nodes where
$p$ is in $P$ and $v$ is in $G$ is a {\em connection for $v$ with
  respect to $P$}.  A set ${\cal C}$ of such connections {\em covers
  $v$ in $G$ with respect to $P$} if, for every node $p$ of $P$,
there is a connection $(p',v)$ in ${\cal C}$ such that
\begin{equation} \label{eq:cover}
\dist(p,p')+\dist(p',v) \leq (1+\epsilon)\, \dist(p,v)
\end{equation}
Let $u,v$ be nodes of the input graph.  Let $Q$ be the shortest
$u$-to-$v$ path that intersects $P$.  Suppose ${\cal C}$ is a set of
connections that covers $u$ and $v$ with respect to $P$.  Then
it contains connections $(p,u), (p',v)$ such that
\begin{equation} \label{eq:connections-and-shortest-path}
\dist(u,p)+\dist(p,p')+\dist(p',v) \leq (1+\epsilon) \length(Q)
\end{equation}
Thorup gives an algorithm that, given a (mostly) planar graph $G$ and
a shortest path $P$, computes a set ${\cal C}$ of connections
that covers all nodes of $G$ and that has $O(\epsilon^{-1})$
connections per node $v$.  In Section~\ref{sec:planar:prepro}, we give
an algorithm that achieves a faster\footnote{Our algorithm depends on
  $G$ being wholly planar (as opposed to mostly planar as in~\cite{ThorupJACM04}, 
  which is the case for the variant we
  address.} running time by covering only a subset of the nodes of
$G$.  The distance oracle involves storing with each node $v$ the
connections that cover $v$ with respect to several shortest paths (and
the distances associated with these connections).  The storage required for $v$
thus has size $O(\epsilon^{-1})$ times the number of such paths.

The second idea is {\em recursively decomposing a planar graph with
  shortest-path separators}.\full{\footnote{A similar but more involved such
  decomposition arose in~\cite{AGKKW98}.}}{} This idea is based on a lemma
in~\cite{LT79} stating that, for any 
spanning tree $T$ in a planar graph in which every face is a triangle,
there is a nontree edge $e$ such that the unique simple cycle in $T
\cup \set{e}$ is a balanced separator.  The nodes of this separator
comprise two paths in $T$.

The distance-oracle construction uses this lemma with $T$ being a
shortest-path tree to recursively decompose the input graph.  The
recursive decomposition defines a binary {\em decomposition} tree in
which each node $x$ is labeled by (i) a subgraph $G(x)$ of the input
graph and (ii) the separator $S(x)$ used to decompose $G(x)$, if $x$
is not a leaf.  If $x$ is the root, $G(x)$ is the input graph.  If
$x$ has children $y$ and $z$, removing the separator $S(x)$ from
$G(x)$ results in two separated subgraphs, $G(y)$ and $G(z)$.  If $x$
is a leaf, $G(x)$ consists of one node.

Each input-graph node $v$ is associated with some decomposition-tree
node, namely, the leafmost node $x$ whose subgraph includes $v$.  We
say that the ancestors of $x$ are {\em relevant} to $v$.  Thus each
input-graph node $v$ has $O(\lg n)$ relevant tree-nodes.  
The distance oracle assigns a label to $v$ that consists of a set of
connections; for each tree-node $x$ relevant to $v$, for each of the
two paths $P$ comprising the separator $S(x)$, the distance oracle
stores a set of connections that cover $v$ in $G(x)$ with respect to $P$.
It follows that the label of $v$ has size $O(\epsilon^{-1} \lg n)$.

Next we show that these labels suffice to estimate point-to-point
distances.  We say that a tree-node $x$ is {\em relevant} to a path
$Q$ if $S(x)$ contains a node of $Q$, and is the {\em most relevant}
if $x$ is the rootmost relevant tree node.
\full{\begin{lemma}}{} If $x$ is the tree-node most relevant to $Q$ then $G(x)$
  contains $Q$.
\full{\end{lemma}}{}
Let $u,v$ be any pair of input-graph nodes, and let $Q$ be the
shortest $u$-to-$v$ path.  Let $x$ be the tree-node most relevant to
$Q$.  Then $G(x)$ contains $Q$, and at least one of the paths
comprising the separator $S(x)$, say $P$, intersects $Q$.  It follows
from~(\ref{eq:connections-and-shortest-path}) that the
$u$-to-$v$ distance is approximately
\begin{equation} \label{eq:dist-est}
\dist(u,p)+\dist(p,p')+\dist(p',v)
\end{equation}
for two nodes $p,p'$ on $P$. 
To estimate the $u$-to-$v$ distance, therefore, the following
procedure suffices: for every tree-node $x$ that is relevant to $u$
and $v$, compute the minimum of~(\ref{eq:dist-est}) over connections
$(p,u)$ and connections $(p',v)$ where $p$ and $p'$ belong to one of the two
paths comprising $S(x)$.  This takes time proportional to the number of such
  connections.  \full{We review this process in
Section~\ref{sec:planar:query} since in our case the situation is
slightly more complicated.}{}

\subsection{Our compact distance oracle}

Our linear-space construction draws on another kind of recursive
decomposition using separators.  
\full{
\begin{definition}[{Frederickson~\cite{journals/siamcomp/Frederickson87}}]
  A {\em division} of a graph $G$ is a partition of the edges of $G$
  into edge-induced subgraphs. A node of $G$ is a {\em boundary node}
  of the partition if it belongs to more than one subgraph. An
  $r$--division of an $n$-node planar graph $G$ is a division of $G$
  into $O(n/r)$ subgraphs, called {\em regions}, with the following
  properties:
(i) Each region contains $O(r)$ edges, and (ii)
the number of boundary nodes in each region is at most $O(\sqrt{r})$.
\label{def:kleinsubramanian}
\end{definition}
\noindent Note that there are $O(n/\sqrt{r})$ boundary nodes in
total.

\begin{lemma}[{Frederickson~\cite{journals/siamcomp/Frederickson87}}]
A planar graph on $n$ vertices can be divided into an
$r$--division in $O(n\lg n)$ time.
\end{lemma}
}{
Frederickson~\cite{journals/siamcomp/Frederickson87} introduced the 
notion of an $r$--division, which is a partition of the edges
  into edge-induced subgraphs (called {\em regions}) such that 
  each region contains $O(r)$ edges and 
the number of boundary nodes in each region is at most $O(\sqrt{r})$.
An $r$--division can be computed in time $O(n\lg n)$. 
}

\full{\subsubsection{Using an $r$--division to obtain linear space}}{}

Before carrying out the recursive decomposition with shortest-path
separators, our preprocessing algorithm computes an $r$--division for $r=\ell^2$
(where $\ell$ is a parameter).  Subsequently, connections $(v,w)$ are
only stored for those nodes $v$ that are boundary nodes of the
$r$--division.  Since there are $O(n/\sqrt{r})$ boundary nodes, the
connections and associated distances require storage
$O((n\epsilon^{-1} \lg n)/\sqrt{r})$.  We choose
$\ell=\Theta(\epsilon^{-1} \lg n)$ so the total storage is $O(n)$.

An {\em $s$-to-$t$ query} is handled as follows.  First, compute
shortest-path distances from $s$ to all the nodes in $s$'s region
$R_s$.  This takes $O(\ell^2)$
time~\cite{journals/jcss/HenzingerKRS97}.  At this point, the query 
algorithm has distances in the subgraph $R_s$ from $s$ to all the
boundary nodes of $R_s$ (and to $t$, if $t$ is in $R_s$).  There are
$O(\ell)$ such boundary nodes.  Similarly, compute shortest-path
distances to $t$ from all the nodes in $t$'s region $R_t$, obtaining
distances in the subgraph $R_t$ to $t$ from all the boundary nodes of
$R_t$.

Let $A, B$ be, respectively, the set of connections for boundary nodes
of $R_s, R_t$.  For each separator path $P$ that has connections in
$A$ and $B$, the procedure described in Section~\ref{sec:planar:query}
finds the
shortest $s$-to-$t$ path that enters $P$ via a connection of $A$ and
leaves $P$ via a connection of $B$.  The time is linear in
the number of such connections (see also~\cite[Section~3.2.2]{ThorupJACM04} 
and \cite[Lemma~3.6]{ThorupJACM04}).  Since each of the $O(\ell)$ boundary
nodes of $R_s$ and $R_t$ has $O(\epsilon^{-1} \lg n)$ connections, the total
time for these computations is $O(\ell \epsilon^{-1} \lg n)$.

Finally, return the minimum overall path-length (including the
$s$-to-$t$ distance within $R_s$, if $t$ belongs to $R_s$).  
The total time for handling the query is $O(\ell^2 + \ell 
\epsilon^{-1}\lg n)$.

\full{\subsubsection{Details of the query algorithm}}{}
\label{sec:planar:query}
We now explain how we find the shortest $s$-to-$t$ path
that enters $P$ via a connection of $A$ and leaves $P$ via a
connection of $B$.  The method is a generalization of that
in~\cite[Sections~3.2.1 and 3.2.2]{ThorupJACM04}.

For each connection $(b, p)$ in $A$, $b$ is a boundary node of $R_s$
and we have $\dist(s,b)$.  For each connection $(b,p)$ in $B$, $b$ is
a boundary node of $R_t$ and we have $\dist(t,b)$.
\full{ }{}
Let $C$ be the sequence of all connections $(s,p)$ and $(t,p)$ in
$A\cup P$, sorted according to the position of $p$ on $P$.
We use the following procedure.
\begin{tabbing}
initialize $m_s, m_t , d:= \infty$\\
initialize $\hat p : = p_0$\\
for each connection $(p, b)$ in $C$ in order,\\
\qquad \= $m_s := m_s + \dist(\hat p, p)$\\
             \> $m_t := m_t + \dist(\hat p, p)$\\
             \> $\hat p := p$\\
             \> if $b$ is a boundary node of $R_s$,\\
             \> \qquad \= $m_s := \min \{m_s, \dist(s,b)+\dist(b,p)\}$\\
            \> if $b$ is a boundary node of $R_t$\\
            \>               \> $m_t := \min \{m_t, \dist(t,b)+\dist(b,p)\}$\\
            \>$d := \min\{d, m_s+m_t\}$\\
return $d$
\end{tabbing}
The procedure requires time $O(\text{number of connections considered})$.
The procedure maintains the invariant that, after a node $\hat p$ of $P$
has been considered in the loop, 
$m_s$ is the length of the shortest $s$-to-$\hat p$ path of the form
that goes via a boundary node $b$ of $R_s$ and a connection $(b,p)$
and then travels along $P$ from $p$ to $\hat p$, 
where $p$ appears before $\hat p$ on $P$.  A similar statement holds for $m_t$.  It
follows that the value $d$ returned by the procedure is the length of
the shortest $s$-to-$t$ path that travels in $R_s$ to a boundary node
$b$ of $R_s$, then goes to $P$ via a
connection for $b$, then travels along $P$ then leaves $P$ via a
connection for a boundary node $b'$ of $R_t$ then travels to $t$
within $R_t$.

\full{
\subsection{Extensions: distance oracles for planar digraphs and reachability oracles}
\label{sec:planar:extensions}
Similar techniques (i.e. storing the connections for a subset of the nodes only) apply to more of  
Thorup's results based on shortest path separators~\cite{ThorupJACM04}. 
Instead of using a more sophisticated preprocessing algorithm that computes 
only the connections for the boundary nodes, we may compute the connections for all the nodes 
(using Thorup's preprocessing algorithms as black boxes) and then store only those 
for boundary nodes. By doing so, we obtain the following linear-space oracles. 
For directed planar graphs, there is a $(1+\epsilon)$--approximate distance oracle with  query time 
$O((\epsilon^{-1}(\lg n)(\lg(nN)))^{2})$, where $N$ denotes the largest 
integer weight. Furthermore, there is a reachability oracle with query
time $O(\lg^2 n)$ 
(using~\cite[Theorem~2.7]{ThorupJACM04}). 

We exploit a similar black-box construction for \minorfr\ graphs in Section~\ref{sec:minorfr}. 
}{}

\full{\section{Preprocessing algorithm for linear-space approximate distance
  oracle for planar graphs}}{\section{Improved preprocessing algorithm}}
\label{sec:planar:prepro}

Thorup's preprocessing algorithm for his undirected construction takes
time $O(n \epsilon^{-2}\lg^3 n)$ (as stated in \cite[Theorem~3.19]{ThorupJACM04}).  We
give a preprocessing scheme for our construction that takes time $O(n
\lg^2 n)$, independent of $\epsilon$.  We give details later, but here 
we observe that the factor $O(\epsilon^{-2} \lg n)$ speedup has three 
sources. \full{(This explanation is aimed at readers familiar with Thorup's paper.)}{}

First, since we are not aiming for a query time of $O(\epsilon^{-1})$,
we can use a simpler preprocessing approach than the one
underlying~\cite[Theorem~3.19]{ThorupJACM04}; we use the approach
that for directed graphs
underlies~\cite[Proposition~3.14]{ThorupJACM04}.  The corresponding
bound for undirected graphs is listed in
Table~\ref{tab:planardoresults} as ``implicit.''

Second, we only need to compute
connections for a small subset of the nodes (the boundary nodes of the
$r$--division).  That in itself does not seem to permit an additional
speedup using Thorup's method since his algorithm depends not on the
number of connections stored but on the sizes of the graphs searched.
Therefore, third, in addition we use another approach to finding connections,
one based on the multiple-source shortest-path (MSSP) algorithm of
Klein~\cite{conf/soda/Klein05} or that of Cabello and 
Chambers~\cite{conf/soda/CabelloC07}.

\begin{tabbing}
{\sc Preprocess}$(G_0)$\\
\qquad \= let $B_0$ be the set of boundary nodes of an
$r$--division~\cite{journals/siamcomp/Frederickson87}\\
\> let $T$ be a shortest-path tree\\
\> compute recursive decomposition based on cycle separators \full{of the form}{} $T\cup \set{e}$ \\
\> for each nonroot node $x$ of recursive-decomposition tree,\\
\> \qquad \= for each path $P_i$ ($i=1,2$) comprising $S(x)$,\\
\>\> \qquad compute connections for nodes of $B_0$ in $G(x)$ with respect to
$P_i$
\end{tabbing}

The last step, computing the connections for nodes of $B_0$ in $G(x)$
with respect to $P_i$, works on a graph $G'(x)$ obtained from $G(x)$
by cutting along $P_i$, duplicating the nodes and edges of $P_i$ and
creating a new face whose boundary consists of the two copies of
$P_i$.  This modification destroys paths that cross $P_i$ but such
paths are not needed since $P_i$ is a shortest path.  It has the
advantage that, for each copy $P$ of $P_i$, in $G'(x)$ all nodes of
$P$ lie on a common face.

For each copy $P$, there is a computation that selects connections
$(p,v)$ for specified nodes $v$ with respect to that copy.  The
computation uses an algorithm called {\sc Path}$(G, B, P)$ that takes
time $O((|G| + \text{number of connections})\lg |G|)$ and selects
$O(\epsilon^{-1})$ connections per node $v\in B$.  Since there are two
copies of two paths comprising $S(x)$, the last step of {\sc
  Preprocess} selects $O(\epsilon^{-1})$ connections per node of $B_0$
in $G(x)$.  Therefore the total number of connections for $B_0$ is
$O(\epsilon^{-1} \lg n)$, and the total time is $O(n \lg^2 n + \abs{B_0}
\epsilon^{-1} \lg n)$, which is $O(n \lg^2 n)$.

Now we describe {\sc Path}$(G, B, P)$.  Let the nodes of $P$ be $p_0
\ldots p_s$.  First the algorithm computes\full{\\}{}
$\full{\text{\hspace{1in}}}{}i(v) = \text{argmin}_i
\dist(p_i, v)$ and
$d_v = \min_i \dist(p_i, v)$.\full{\\}{}
These can be computed using a single-source shortest-path computation
in the graph obtained by zeroing out the lengths of the edges of $P$.

For $i=0, 1, \ldots, s$, let $T_i$ denote the shortest-path tree
rooted at $p_i$.  For $i>0$, let $T_i'$ be the tree obtained from
$T_{i-1}$ by removing the parent edge of $p_i$ and adding the edge
$p_i p_{i-1}$, obtaining a $p_i$--rooted tree (not a shortest-path tree).
For $i>0$, let $\sigma_i$ denote a sequence of edges whose insertion into $T'_i$
(followed by the ejection of each corresponding parent edge) result in $T_i$.
\full{ }{}
Klein~\cite{conf/soda/Klein05} shows that each edge is inserted at
most once, and gives an $O(\abs{G}\lg \abs{G})$ algorithm (the {\em
  multiple-source shortest-path algorithm}) to compute these
sequences.  For each such inserted edge $uv$, the algorithm also
computes the resulting change $\Delta_{uv}$ in the length of the
root-to-$v$ path in the tree.
Cabello and Chambers~\cite{conf/soda/CabelloC07} give a simplification
of the multiple-source shortest-path algorithm and generalize it to
bounded-genus in $O(g^2 \abs{G}\lg \abs{G})$ time.
\full{ }{}
The algorithm {\sc Path} uses one of these algorithms to compute the
sequences $\sigma_i$ and the corresponding length changes~$\Delta_{uv}$.

\full{\subsection{The two phases}}{}
The remainder of {\sc Path} consists of two phases, {\sc
  Forward} and {\sc Backward}.
 A connection $(p_i, v)$ might be
added by {\sc Forward} if $i> i(x)$ and by {\sc Backward} if $i <
i(x)$.  We describe {\sc Forward}. {\sc Backward} is symmetric.

\full{\subsection{The {\sc Forward} phase}}{}
The algorithm {\sc Forward} iterates through the nodes $p_0,
\ldots, p_s$ of $P$, maintaining a tree $T$ that is, in turn, $T_0,
T_1, \ldots, T_s$.  The tree $T$ is represented using a dynamic-tree
data structure~\cite{ABHVW04,AHLT05,Frederickson97,ST83,TW05}.  
A node-labeling is maintained:  $\used(v)$ is a quantity
(discussed later) that is
used to decide whether $v$ needs a new connection.  This labeling is
represented implicitly, as is typical in dynamic trees, so as to
support bulk updates.  In this case (somewhat atypically), an update
takes the form ``add a quantity $\Delta$ to the label of every tree in
the subtree rooted at $u$.''  Each update takes $O(\lg n)$ amortized
time.  In addition, searching for a node $v$ that has $\used(v)\leq
0$ takes $O(\lg n)$ time.  

\begin{tabbing}
{\sc Forward}$(G, B, P)$:\\
 \qquad \= initialize $T := T_0$\\
   \> for every node $v$, initialize $\used(v) := \infty$\\
 \> for $i=0, 1,2,\ldots, s$:\\
   \> \qquad \= {\em comment:} $T$ is rooted at $p_i$\\
$\star$ \> \> for each node $v\in B$ such that either $i(v)=i$ or
$\used(v)\leq 0$ ,\\
 \> \> \qquad \= create a connection $(p_i, v)$\\
 \> \> \> set $\used(v) := \epsilon\, d_v$\\
\> \> if $i<s$,\\
 \> \> \>{\em comment:} now change the root...\\
     \> \> \>remove parent edge of $p_{i+1}$ and add edge $p_{i+1} p_i$\\
   \> \> \> {\em comment:} now make the tree a shortest-path tree\\
$\dagger$   \> \> \> for each edge $uv$ in the sequence $\sigma_{i+1}$,\\
    \> \> \> \qquad\= remove the current parent edge of $v$ in $T$, and add $uv$\\
$\ddagger$    \> \> \> \> \= for every active node $w$ in the $v$--rooted subtree of
    $T$,\\
\>\>\>\>\> $\used(w) := \used(w) + \Delta_{uv}$
\end{tabbing}

The overall number of iterations of the loop in Step~$\star$ is the number of
connections added.  The overall number of iterations of the loop in
Step~$\dagger$ is at most the number of edges, which is $O(\abs{G})$.
Step~$\ddagger$ can be done using a single bulk update in $O(\lg \abs{G})$
time.  Consequently, the algorithm runs in time $O((\abs{G}+\text{number of
  connections})\lg \abs{G})$.

Now we show that the algorithm selects a covering set of connections
(and that the set is small).  At each moment in the execution of the
algorithm, for each node $v$ such that $\used(v)$ is finite, let
$\last(v)$ denote the node $p$ of $P$ such that $(p,v)$ was the most
recently selected connection for $v$.

The {\em $\used$ invariant} is: for every node $v$ for which
$\used(v)$ is finite,
\begin{equation} \label{eq:used-invariant}
\used(v) = \epsilon d_v - (\dist(p_i, \last(v))+\dist(\last(v), v) - \dist_T(p_i,v) )
\end{equation}
Note that $\dist(p_i, \last(v)) + \dist(\last(v), v)$ is the length of
the path that goes from the current root $p_i$ to $v$ via $\last(v)$.  When
this length becomes significantly longer than $\dist(p, v)$ (longer by $\epsilon d_v$), 
$\used(v) \leq 0$ so the node $v$ is included in the loop in
Step~$\star$, so the connection $(p_i,v)$ is added.  This shows that the
connections added by {\sc Forward} and {\sc Backward} cover each node
$v\in B$.

To bound the number of connections, we follow Thorup in using the
potential function $\Phi_v = \dist(p_s, \last(v)) + \dist(\last(v),
v)$.  Suppose that, at some execution of Step~$\star$, $\used(v)\leq
0$, so $\dist(p_i, \last(v))+\dist(\last(v), v) - \dist_T(p_i,v)\geq
\epsilon d_v$.  When a connection $(p_i, v)$ is then added, $\last(v)$
becomes $p_i$, so the potential
function $\Phi_v$ is reduced by at least $\epsilon d_v$.

Initially $\Phi_v = \dist(p_s, p_{i(v)}) + \dist(p_{i(v)},
v)$.  Throughout the phase, by the triangle inequality , $\Phi_v \geq
\dist(p_s, v)$.  Again using the triangle inequality (and the fact
that the graph is undirected), $\dist(p_s, p_{i(v)}) \leq \dist(p_s,
v) + \dist(p_{i(v)}, v)$, so $\Phi_v \geq \dist(p_s, v) \geq \dist(p_s, p_{i(v)})
- \dist(p_s, v)$.  Thus the total amount of reduction in $\Phi_v$ is
at most $2\,\dist(p_s, v)$.  Since each reduction is by at least
$\epsilon \, \dist(r_{i(v)}, v)$, the total number of reductions
(number of connections added by {\sc Forward} after the initial one) is at most
$\ceil{2\epsilon^{-1}}$.

\section{Approximate distance oracles for genus~$g$ graphs}
\label{sec:genuslabels}

\begin{theorem}
\genusthm
\label{thm:genusthm}
\end{theorem}

Our distance oracle for genus~$g$ graphs is based on separating shortest paths, 
as for planar graphs (see Section~\ref{sec:planar:overview}). Thorup~\cite{ThorupJACM04}
proves that any planar graph can be recursively separated by three shortest paths. 
\full{
Abraham and Gavoille~\cite{AbrahamG06} extend his result to minor-closed families, 
proving that any \minorfr\ graph can be recursively separated by $O(1)$ shortest paths. 
Since bounded-genus graphs exclude minors, we could use their result to obtain a 
linear-space distance oracle. The constant in~\cite{AbrahamG06} however depends 
on the size of the minor in an unspecified way. 
}{}
In the following, we prove that genus~$g$ graphs can be recursively separated using 
at most $O(g)$ shortest paths. In fact, only the first separator consists of at most $2g$ 
paths, while lower levels can be separated using $3$ paths. 
These smaller separators allow us to derive approximate oracles and
labeling schemes with a dependency on $g$ that is much lower than the corresponding 
dependency in the more general construction by Abraham and Gavoille~\cite{AbrahamG06}. 
More formally, we also prove the following. 
\begin{theorem}[fast distance queries for genus~$g$ graphs]
\genusthmfastquery
\label{thm:genusfastquerythm}
\end{theorem}

\full{In the following, we assume that $G$ is embedded. }{}

\full{\subsection{Overview}}{\paragraph{Overview}}
In the first step, we ``cut'' the genus~$g$ graph into planar subgraphs using the 
{\em tree-cotree decomposition} of Eppstein~\cite{conf/soda/Eppstein03}, which
decomposes a graph of genus~$g$ into 
planar graphs, separated by $2g$ paths from a tree $T$.   We choose
$T$ to be a shortest-path tree.
\full{
\begin{lemma}[{Corollary of Eppstein~\cite[Proof of Lemma~3.2]{conf/soda/Eppstein03}}] 
Any graph $G$ of genus~$g$ on $n$ nodes and $m$ edges can be divided into 
planar subgraphs by a separator that consists of at most $2g$ shortest paths. 
Furthermore, these paths can be computed using a single-source shortest-path search in $G$ 
plus $O(gm)$ time. 
\end{lemma}
\begin{proof}
The existence of these paths is due to Eppstein's {\em tree-cotree decomposition} $(T,C,X)$ of $G$, 
where $T$ is a SSSP tree (any spanning tree works), $C$ is its {\em cotree}, and $X$ is the set of 
remaining edges of size at most $2g$~\cite{conf/soda/Eppstein03}. 

Starting at an arbitrary node $u$, we compute a single-source shortest path tree $T$ and its dual $C$. 
The separator consists of at most $2g$ tree paths from $u$ to $v_i$ for each of the remaining edges 
$(v_0,v_1)\in X$, and of the edges in $X$. 
\qed
\end{proof}
}{}
At a high level, the theorems follow by combining Eppstein's lemma\full{}{~\cite[Proof of Lemma~3.2]{conf/soda/Eppstein03}} with the distance oracles for planar graphs 
(Thorup~\cite{ThorupJACM04} and Sections~\ref{sec:planar} and \ref{sec:planar:prepro}). 
Within the planar subgraphs, we use the distance oracles for planar graphs.  In addition to computing the 
connections to the separator paths within each planar subgraph, we also need to compute the connections 
to the $O(g)$ tree-cotree decomposition paths. Note that the latter set of connections consists of paths that 
may pass through non-planar parts. To compute these, we may use either~\cite{conf/soda/CabelloC07} 
or~\cite[Lemma~3.12]{ThorupJACM04}, depending on the values of $g$ and $\epsilon$. 

\full{\subsection{Preprocessing algorithm}\label{sec:genus:prepro}}{\paragraph{Preprocessing algorithm}}

\full{\paragraph{Within planar subgraphs}}{}
To obtain the preprocessing and space bounds in Theorem~\ref{thm:genusthm}, 
we use the preprocessing algorithm described in Section~\ref{sec:planar:prepro} 
with $r:=\ell^2$, where $\ell=O(\epsilon^{-1}(\lg n + g))$. Since the 
number of connections per node is proportional to $\ell$ and since a 
$1/\sqrt r$--fraction of the nodes per subgraph lies on the boundary, the overall 
space consumption is linear. 
\full{ }{}
To obtain the preprocessing and space bounds in Theorem~\ref{thm:genusfastquerythm}, 
we use Thorup's algorithm~\cite[Thm.~3.19]{ThorupJACM04} for $O(1/\epsilon)$ query time. 

\full{\paragraph{Connections to tree-cotree separator}}{}
There are two options to compute \full{these connections}{connections to the tree-cotree separator}: 
(1) We may use~\cite[Lemma~3.12]{ThorupJACM04} (which internally uses Thorup's $O(m)$ 
SSSP algorithm~\cite{journals/jacm/Thorup99,journals/jal/Thorup00}). 
The lemma states that, for a path $Q$, we can compute an $\epsilon$--covering set $\mathcal C(v,Q)$ 
for all nodes $v$ in time $O(\epsilon^{-1}n(\lg n))$. 
(2) We may use the MSSP data structure for genus~$g$ graphs by Cabello and 
Chambers~\cite{conf/soda/CabelloC07}, which requires $O(g^2n\lg n)$ preprocessing 
and then answers queries in time $O(\lg n)$. See planar preprocessing (Section~\ref{sec:planar:prepro}) 
for details. The time required is $O(g^2n\lg n+\text{number of connections}\cdot\lg n)=O(g^2n\lg n)$. 
\full{ }{}
We apply either lemma for the at most $2g$ paths of the tree-cotree decomposition. (For 
Theorem~\ref{thm:genusfastquerythm}, the first option gives faster asymptotic preprocessing time; 
for Theorem~\ref{thm:genusthm}, the optimal choice depends on $\epsilon$ and~$g$.)

\full{\subsection{Query algorithm}}{\paragraph{Query algorithm}}
At query time, we can essentially use the same algorithm as for the planar case (Section~\ref{sec:planar:query} 
and \cite[Thm.~3.19]{ThorupJACM04}). 
The only difference to the planar case is that we also need to include the at most $2g$ paths separating the genus graph 
into planar subgraphs. 
To obtain the bound on the query time in Theorem~\ref{thm:genusthm}, note that computing 
connections through these $\leq2g$ separating paths can be done in time $O(\ell g/\epsilon)$ and 
that exploring both regions took time $O(\ell^2)$ (where $\ell=O(\epsilon^{-1}(\lg n + g))$).

\full{
\section{Linear-space approximate distance oracle for $H$--\minorfr\ graphs}
\label{sec:minorfr}

\begin{theorem}
\minorthm
\label{thm:minorthm}
\end{theorem}

The proof of Theorem~\ref{thm:minorthm} is structurally the same as for the planar case. 
We again use $r$--divisions~\cite{journals/siamcomp/Frederickson87}, this time 
tailored to \minorfr\ graphs using a separator algorithm by Kawarabayashi and 
Reed~\cite{conf/focs/KawarabayashiR10}\footnote{Alternatively, Frederickson's 
algorithm could be combined with the separator algorithm by 
Alon, Seymour, and Thomas~\cite{AlonST94} to obtain an $r$--division with 
$O(\abs H^{3/2}\sqrt r)$ boundary vertices per region or with the separator algorithm by 
Reed and Wood~\cite{journals/talg/ReedW09} for  
$O(\abs H^{3/2}2^{(\abs H^2+4)/2}r^{2/3})$ boundary vertices per region. 
See also~\cite{journals/dam/TazariM09} for such modifications. 
}. 

An {\em $r$--division} of an  $H$--minor-free graph is a division
into $\Theta(n/r)$ regions of $O(r)$ vertices each and $O(\abs H\sqrt
r)$ boundary vertices each. (Note that the boundary is larger by a factor $\abs H$ 
compared to the boundary in the planar case.)
\begin{lemma}[{Frederickson~\cite{journals/siamcomp/Frederickson87}, combined with Kawarabayashi and Reed~\cite{conf/focs/KawarabayashiR10}}]
An $H$--minor-free graph on $n$ vertices can be divided into an
$r$--division in $O(n^2\lg n)$ time.
\label{lemma:minordivision}
\end{lemma}

\begin{proof}[of Theorem~\ref{thm:minorthm}]
The proof is a combination of three techniques: (i) Lemma~\ref{lemma:minordivision}, (ii) shortest 
path separators in \minorfr\ graphs by Abraham and Gavoille~\cite{AbrahamG06}, and (iii) Thorup's 
$\epsilon$--covers (as described in Section~\ref{sec:planar:overview}). 

Abraham and Gavoille~\cite{AbrahamG06} prove that any $H$--\minorfr\ graph can be recursively 
separated by $k(H)$ shortest paths and that these paths can be found in polynomial
time. We first compute these separator paths. 
We then compute an $r$--division for $r:=(\ell\cdot\abs H)^2$ as in Lemma~\ref{lemma:minordivision}. 
For all the nodes on the boundary, we compute connections to these shortest 
paths~\cite[Lemma~3.12]{ThorupJACM04} (as in Section~\ref{sec:genus:prepro}). 
We store $\ell=O(\epsilon^{-1}\lg n)$ connections per node (where $\ell$ 
depends on $\abs H$ in an unspecified way) for $O(n/\ell)$ nodes. 
The total space requirement is thus $O(m)$ (without 
further dependencies on $H$). 

At query time, given a pair $(s,t)$, we first explore both regions $R_s$ and $R_t$, respectively, 
using an SSSP search~\cite{journals/dam/TazariM09} until all the boundary nodes of these 
two regions have been found. This step requires time $O(\ell^2)$ (hiding 
further dependencies on $H$ stemming from~\cite{journals/dam/TazariM09}). 
As described in Section~\ref{sec:planar:query}, we can then merge clean and ordered covers 
in linear time~\cite[Lemma~3.6]{ThorupJACM04}. 
\qed
\end{proof}
}{}

\full{
\section{Linear-space approximate distance oracle for unit-length graphs with bounded doubling dimension}
\label{sec:doubling}

The distance oracles for planar, bounded-genus, and minor-free graphs heavily used the 
notion of {\em separators}. In the following, we show that separators are not the only way 
to obtain linear-space approximate distance oracles. 
Our linear-space approximate distance oracle for bounded-doubling-dimension graphs 
exploits the bounded-growth property. 

The {\em aspect ratio}
(also known as {\em spread}) of $P\subseteq V$, denoted by
$\Delta(P)$, is the ratio of the diameter of $P$ and the distance
between the closest pair of nodes in $P$.
It is well known that for any $\lambda$--doubling metric $\mathcal M$,
any set of points $P\subset\mathcal M$ with aspect ratio at most
$D\geq\Delta(P)$ satisfies $\abs{P}\leq\lambda^{O(\lg D)}$.

We use this fact to obtain the following. 

\begin{theorem}
\doublingcor
\label{thm:doublingthm}
\end{theorem}
Note that (1) for unit-length graphs $\Delta=O(n)$ and thus the query time 
for constant $\alpha$ is $O(poly(\lg n,1/\epsilon))$ and (2)
our result also holds for unweighted geometric graphs such 
as those considered in~\cite{GudmundssonLevcopoulosNarasimhanSmid}. 
Our approach extends to graphs with moderate edge lengths;
the dependency on the largest weight is however polynomial and not logarithmic.

In our proof we use the following approximate distance labeling scheme.

\begin{lemma}[{Har-Peled and Mendel~\cite[Proposition~6.10]{MendelHarPeled}}]
For a metric with doubling dimension $\alpha$ and aspect ratio $\Delta=\Delta(V)$
and for any $\epsilon > 0$, there exists a $(1 +
\epsilon)$--approximate distance labeling scheme with label length
$(1/\epsilon)^{O(\alpha)}\lg\Delta$~bits per node and query time $2^{O(\alpha)}$. 
\label{lemma:doublinglabel}
\end{lemma}

\begin{proof}[of Theorem~\ref{thm:doublingthm}]
Let $\delta < n$ be an integer.  A {\em $\delta$--dominating set} of a graph
$G = (V,E)$ is a subset $L \subseteq V$ of nodes such that for each
$v\in V$ there is a node $l\in L$ at distance at most $\delta$.  It is
well-known that there is a $\delta$--dominating set $L$ of size at
most $\abs{L}\leq n/(\delta+1)$ and that such a set $L$ can be found
efficiently~\cite{KP98}.

Let $G=(V,E)$ be a graph that allows for
a $(1+\epsilon)$--approximate distance labeling scheme with 
label length $\ell=(1/\epsilon)^{O(\alpha)}\lg\Delta$ (Lemma~\ref{lemma:doublinglabel}). 
We store the distance labels for 
 the nodes of an $\ell$--dominating set, which requires total
space $O(n)$. For each unlabeled node $v$, we also store its 
nearest labeled node $l(v)$. 

At query time, given $s,t\in V$, we distinguish between
``close'' pairs (distance at most $\ell/\epsilon$) and ``far'' pairs (otherwise). 
For ``close'' pairs, we explore (using BFS) the ball $B(s)$ of radius $d=O(\ell/\epsilon)$ 
around $s$. If $t\in B(s)$, the exact distance can be returned (a ``close'' pair).
Recall that for any $\lambda$--doubling metric $\mathcal M$,
any set of points $P\subset\mathcal M$ with aspect ratio at most
$D\geq\Delta(P)$ satisfies $\abs{P}\leq\lambda^{O(\lg D)}$.
For a unit-length graph $G=(V,E)$ with doubling dimension
$\alpha=\lg_2\lambda$, for any node $v\in V$ the number of nodes within
distance $d$ satisfies
\begin{displaymath}
\abs{\{u:d_G(u,v)\leq d\}} \leq \lambda^{O(\lg d)}. 
\end{displaymath}
The number of edges within $B(s)$ is at most quadratic in the number of nodes. 
Exploring $B(s)$ using BFS thus requires time proportional to $\lambda^{O(\lg d)}$.
For ``far'' pairs, we triangulate using $l(s)$ and $l(t)$, returning an approximate 
distance: The algorithm returns 
$\tilde d(u,v) = d_G(u,l(u)) + \mathcal D(\mathcal L(l(u)),\mathcal L(l(v))) + d_G(l(v),v)$, 
where $\mathcal L(w)$ denotes the label of $w$ 
and $\mathcal D(\cdot,\cdot)$ denotes the decoding function of the labeling scheme. 
A simple calculation (using the triangle inequality and the fact that the distance $d(u,v)$ 
is at least $\ell/\epsilon$) yields that the query result 
$\tilde d(u,v)$ satisfies $\tilde d(u,v)\leq (1+7\epsilon)d(u,v)$ (for any $\epsilon\in(0,1]$). 
\qed
\end{proof}

}

\full{\newpage}

\bibliographystyle{alpha}
\bibliography{planaradq}

\newcommand{\etalchar}[1]{$^{#1}$}
\begin{thebibliography}{AHdLT05}

\bibitem[ABH{\etalchar{+}}04]{ABHVW04}
Umut~A. Acar, Guy~E. Blelloch, Robert Harper, Jorge~L. Vittes, and Shan
  Leung~Maverick Woo.
\newblock Dynamizing static algorithms, with applications to dynamic trees and
  history independence.
\newblock In {\em Proceedings of the Fifteenth ACM-SIAM Symposium on Discrete
  Algorithms}, pages 531--540, 2004.

\bibitem[AG06]{AbrahamG06}
Ittai Abraham and Cyril Gavoille.
\newblock Object location using path separators.
\newblock In {\em Proceedings of the Twenty-Fifth Annual ACM Symposium on
  Principles of Distributed Computing (PODC)}, pages 188--197, 2006.
\newblock Details in LaBRI Research Report RR-1394-06.

\bibitem[AGK{\etalchar{+}}98]{AGKKW98}
Sanjeev Arora, Michelangelo Grigni, David~R. Karger, Philip~N. Klein, and
  Andrzej Woloszyn.
\newblock A polynomial-time approximation scheme for weighted planar graph
  {TSP}.
\newblock In {\em Proceedings of the Ninth Annual ACM-SIAM Symposium on
  Discrete Algorithms (SODA)}, pages 33--41, 1998.

\bibitem[AHdLT05]{AHLT05}
Stephen Alstrup, Jacob Holm, Kristian de~Lichtenberg, and Mikkel Thorup.
\newblock Maintaining information in fully dynamic trees with top trees.
\newblock {\em ACM Transactions on Algorithms}, 1(2):243--264, 2005.

\bibitem[AST94]{AlonST94}
Noga Alon, Paul~D. Seymour, and Robin Thomas.
\newblock Planar separators.
\newblock {\em SIAM Journal on Discrete Mathematics}, 7(2):184--193, 1994.

\bibitem[BGK{\etalchar{+}}10]{journals/corr/abs-1008-1480}
Yair Bartal, Lee-Ad Gottlieb, Tsvi Kopelowitz, Moshe Lewenstein, and Liam
  Roditty.
\newblock Fast, precise and dynamic distance queries.
\newblock {\em CoRR}, abs/1008.1480, 2010.
\newblock To appear in SODA 2011.

\bibitem[Cab06]{conf/soda/Cabello06}
Sergio Cabello.
\newblock Many distances in planar graphs.
\newblock In {\em Proceedings of the Seventeenth Annual ACM-SIAM Symposium on
  Discrete Algorithms (SODA)}, pages 1213--1220, 2006.
\newblock A preprint of the journal version is available in the University of
  Ljubljana preprint series, Vol. 47 (2009), 1089.

\bibitem[CC07]{conf/soda/CabelloC07}
Sergio Cabello and Erin~W. Chambers.
\newblock Multiple source shortest paths in a genus $g$ graph.
\newblock In {\em Proceedings of the Eighteenth Annual ACM-SIAM Symposium on
  Discrete Algorithms, SODA 2007, New Orleans, Louisiana, USA}, pages 89--97,
  2007.

\bibitem[CX00]{stoc/ChenX00}
Danny~Z. Chen and Jinhui Xu.
\newblock Shortest path queries in planar graphs.
\newblock In {\em Proceedings of the ACM Symposium on Theory of Computing
  (STOC)}, pages 469--478, 2000.

\bibitem[CZ00]{journals/algorithmica/ChaudhuriZ00}
Shiva Chaudhuri and Christos~D. Zaroliagis.
\newblock Shortest paths in digraphs of small treewidth. part {I}: Sequential
  algorithms.
\newblock {\em Algorithmica}, 27(3):212--226, 2000.
\newblock Announced at ICALP 1995.

\bibitem[DPZ91]{conf/icalp/DjidjevPZ91}
Hristo Djidjev, Grammati~E. Pantziou, and Christos~D. Zaroliagis.
\newblock Computing shortest paths and distances in planar graphs.
\newblock In {\em Automata, Languages and Programming, 18th International
  Colloquium, ICALP91, Madrid, Spain, July 8-12, 1991, Proceedings}, pages
  327--338, 1991.

\bibitem[DPZ95]{conf/stacs/DjidjevPZ95}
Hristo Djidjev, Grammati~E. Pantziou, and Christos~D. Zaroliagis.
\newblock On-line and dynamic algorithms for shorted path problems.
\newblock In {\em STACS}, pages 193--204, 1995.

\bibitem[DPZ00]{journals/algorithmica/DjidjevPZ00}
Hristo Djidjev, Grammati~E. Pantziou, and Christos~D. Zaroliagis.
\newblock Improved algorithms for dynamic shortest paths.
\newblock {\em Algorithmica}, 28(4):367--389, 2000.

\bibitem[Epp03]{conf/soda/Eppstein03}
David Eppstein.
\newblock Dynamic generators of topologically embedded graphs.
\newblock In {\em Proceedings of the Fourteenth Annual ACM-SIAM Symposium on
  Discrete Algorithms (SODA)}, pages 599--608, 2003.

\bibitem[FR06]{journals/jcss/FakcharoenpholR06}
Jittat Fakcharoenphol and Satish Rao.
\newblock Planar graphs, negative weight edges, shortest paths, and near linear
  time.
\newblock {\em Journal of Computer and System Sciences}, 72(5):868--889, 2006.
\newblock Announced at FOCS 2001.

\bibitem[Fre87]{journals/siamcomp/Frederickson87}
Greg~N. Frederickson.
\newblock Fast algorithms for shortest paths in planar graphs, with
  applications.
\newblock {\em SIAM Journal on Computing}, 16(6):1004--1022, 1987.

\bibitem[Fre97]{Frederickson97}
Greg~N. Frederickson.
\newblock A data structure for dynamically maintaining rooted trees.
\newblock {\em Journal of Algorithms}, 24:37--65, 1997.
\newblock Announced at SODA 1993.

\bibitem[GLNS08]{GudmundssonLevcopoulosNarasimhanSmid}
Joachim Gudmundsson, Christos Levcopoulos, Giri Narasimhan, and Michiel H.~M.
  Smid.
\newblock Approximate distance oracles for geometric spanners.
\newblock {\em ACM Transactions on Algorithms}, 4(1), 2008.
\newblock Announced at SODA and ISAAC 2002.

\bibitem[GPPR04]{journals/jal/GavoillePPR04}
Cyril Gavoille, David Peleg, St{\'e}phane P{\'e}rennes, and Ran Raz.
\newblock Distance labeling in graphs.
\newblock {\em J. Algorithms}, 53(1):85--112, 2004.
\newblock Announced at SODA 2001.

\bibitem[HKRS97]{journals/jcss/HenzingerKRS97}
Monika~Rauch Henzinger, Philip~Nathan Klein, Satish Rao, and Sairam
  Subramanian.
\newblock Faster shortest-path algorithms for planar graphs.
\newblock {\em Journal of Computer and System Sciences}, 55(1):3--23, 1997.
\newblock Announced at STOC 1994.

\bibitem[HPM06]{MendelHarPeled}
Sariel Har-Peled and Manor Mendel.
\newblock Fast construction of nets in low dimensional metrics, and their
  applications.
\newblock {\em SIAM J. Comput.}, 35(5):1148--1184, 2006.
\newblock Announced at SOCG 2005.

\bibitem[KK06]{journals/talg/KowalikK06}
Lukasz Kowalik and Maciej Kurowski.
\newblock Oracles for bounded-length shortest paths in planar graphs.
\newblock {\em ACM Transactions on Algorithms}, 2(3):335--363, 2006.
\newblock Announced at STOC 2003.

\bibitem[KKS11]{TR}
Kenichi Kawarabayashi, Philip~Nathan Klein, and Christian Sommer.
\newblock Linear-space approximate distance oracles for planar, bounded-genus,
  and minor-free graphs.
\newblock {\em CoRR}, abs/????.????, 2011.

\bibitem[Kle05]{conf/soda/Klein05}
Philip~Nathan Klein.
\newblock Multiple-source shortest paths in planar graphs.
\newblock In {\em Proceedings of the Sixteenth Annual ACM-SIAM Symposium on
  Discrete Algorithms (SODA)}, pages 146--155, 2005.

\bibitem[KMW10]{journals/talg/KleinMW10}
Philip~Nathan Klein, Shay Mozes, and Oren Weimann.
\newblock Shortest paths in directed planar graphs with negative lengths: A
  linear-space {$O(n\log^2n)$}-time algorithm.
\newblock {\em ACM Transactions on Algorithms}, 6(2), 2010.
\newblock Announced at SODA 2009.

\bibitem[KP98]{KP98}
Shay Kutten and David Peleg.
\newblock Fast distributed construction of small $k$-dominating sets and
  applications.
\newblock {\em Journal of Algorithms}, 28(1):40--66, 1998.
\newblock Announced at PODC 1995.

\bibitem[KR10]{conf/focs/KawarabayashiR10}
Kenichi Kawarabayashi and Bruce~A. Reed.
\newblock A separator theorem in minor-closed classes.
\newblock In {\em 51st Annual IEEE Symposium on Foundations of Computer
  Science, FOCS 2010}, 2010.

\bibitem[LT79]{LT79}
Richard~J. Lipton and Robert~Endre Tarjan.
\newblock A separator theorem for planar graphs.
\newblock {\em SIAM Journal on Applied Mathematics}, 36(2):177--189, 1979.

\bibitem[MN07]{MendelN07}
Manor Mendel and Assaf Naor.
\newblock Ramsey partitions and proximity data structures.
\newblock {\em Journal of the European Mathematical Society}, 9(2):253--275,
  2007.
\newblock Announced at FOCS 2006.

\bibitem[MS10]{journals/corr/abs-1011-5549}
Shay Mozes and Christian Sommer.
\newblock Exact shortest path queries for planar graphs using linear space.
\newblock {\em CoRR}, abs/1011.5549, 2010.

\bibitem[MWN10]{esa/MozesW10}
Shay Mozes and Christian Wulff-Nilsen.
\newblock Shortest paths in planar graphs with real lengths in
  {$O(n\log^2n/\log\log n)$} time.
\newblock In {\em Algorithms - ESA 2010, 18th Annual European Symposium}, 2010.

\bibitem[MZ07]{distance-oracle-experiment}
Laurent~Flindt Muller and Martin Zachariasen.
\newblock Fast and compact oracles for approximate distances in planar graphs.
\newblock In {\em Proceedings of the 15th annual European Conference on
  Algorithms}, pages 657--668, 2007.

\bibitem[Nus10]{Nussbaum10}
Yahav Nussbaum.
\newblock Improved distance queries in planar graphs.
\newblock {\em CoRR}, abs/1012.2825, 2010.

\bibitem[PR10]{PatrascuRoditty}
Mihai Patrascu and Liam Roditty.
\newblock Distance oracles beyond the {Thorup--Zwick} bound.
\newblock In {\em 51st Annual IEEE Symposium on Foundations of Computer Science
  (FOCS)}, 2010.

\bibitem[RW09]{journals/talg/ReedW09}
Bruce~A. Reed and David~R. Wood.
\newblock A linear-time algorithm to find a separator in a graph excluding a
  minor.
\newblock {\em ACM Transactions on Algorithms}, 5(4), 2009.

\bibitem[Sli07]{SlivkinsPODC05JournalVersion}
Aleksandrs Slivkins.
\newblock Distance estimation and object location via rings of neighbors.
\newblock {\em Distributed Computing}, 19(4):313--333, 2007.
\newblock Announced at PODC 2005.

\bibitem[ST83]{ST83}
Daniel~Dominic Sleator and Robert~Endre Tarjan.
\newblock A data structure for dynamic trees.
\newblock {\em Journal of Computer and System Sciences}, 26(3):362--391, 1983.
\newblock Announced at STOC 1981.

\bibitem[SVY09]{SparseDO}
Christian Sommer, Elad Verbin, and Wei Yu.
\newblock Distance oracles for sparse graphs.
\newblock In {\em 50th Annual IEEE Symposium on Foundations of Computer Science
  (FOCS)}, pages 703--712, 2009.

\bibitem[Tal04]{conf/stoc/Talwar04}
Kunal Talwar.
\newblock Bypassing the embedding: algorithms for low dimensional metrics.
\newblock In {\em Proceedings of the 36th Annual ACM Symposium on Theory of
  Computing (STOC)}, pages 281--290, 2004.

\bibitem[Tho99]{journals/jacm/Thorup99}
Mikkel Thorup.
\newblock Undirected single-source shortest paths with positive integer weights
  in linear time.
\newblock {\em Journal of the ACM}, 46(3):362--394, 1999.
\newblock Announced at FOCS 1997.

\bibitem[Tho00]{journals/jal/Thorup00}
Mikkel Thorup.
\newblock Floats, integers, and single source shortest paths.
\newblock {\em Journal of Algorithms}, 35(2):189--201, 2000.
\newblock Announced at STACS 1998.

\bibitem[Tho04]{ThorupJACM04}
Mikkel Thorup.
\newblock Compact oracles for reachability and approximate distances in planar
  digraphs.
\newblock {\em Journal of the ACM}, 51(6):993--1024, 2004.
\newblock Announced at FOCS 2001.

\bibitem[TMH09]{journals/dam/TazariM09}
Siamak Tazari and Matthias M{\"u}ller-Hannemann.
\newblock Shortest paths in linear time on minor-closed graph classes, with an
  application to {Steiner} tree approximation.
\newblock {\em Discrete Applied Mathematics}, 157(4):673--684, 2009.
\newblock Announced at WG 2008.

\bibitem[TW05]{TW05}
Robert~Endre Tarjan and Renato Fonseca~F. Werneck.
\newblock Self-adjusting top trees.
\newblock In {\em Proceedings of the Sixteenth Annual ACM-SIAM Symposium on
  Discrete Algorithms}, pages 813--822, 2005.

\bibitem[TZ05]{ThorupZwick2005}
Mikkel Thorup and Uri Zwick.
\newblock Approximate distance oracles.
\newblock {\em Journal of the ACM}, 52(1):1--24, 2005.
\newblock Announced at STOC 2001.

\end{thebibliography}

\end{document}